\documentclass[letterpaper, 10 pt, conference]{ieeeconf}

\usepackage{amsfonts,amsmath,amssymb} 

\DeclareMathOperator{\diag}{diag}         
\def\rbb{\mathbb{R}}
\def\trp{^T}
\def\diag{{\rm diag}}

\def\half{\frac{1}{2}}
\usepackage{psfrag,color}
\usepackage{enumerate,cite,latexsym,graphicx,subfig}
\newtheorem{theorem}{Theorem}
\newtheorem{lemma}{Lemma}

\newtheorem{definition}{Definition}

\newtheorem{remark}{Remark}

\title{\LARGE \bf Time Averaged Consensus in a Direct Coupled Coherent Quantum Observer Network for a Single Qubit Finite Level Quantum System}

\author{Ian R.~Petersen %
\thanks{This work was supported by the
Australian Research Council (ARC) and the Air Force Office of Scientific
Research (AFOSR). This material is based on research sponsored by the
Air Force Research Laboratory, under agreement number FA2386-12-1-4075.  The U.S. Government is authorized to reproduce and
distribute reprints for Governmental purposes notwithstanding any
copyright notation thereon.
The views and conclusions contained herein are those of the authors
and should not be interpreted as necessarily representing the official
policies or endorsements, either expressed or implied, of the Air
Force Research Laboratory or the U.S. Government. }%
\thanks{Ian R. Petersen is with the School of  Engineering and Information Technology, 
        University of New South Wales at the Australian Defence Force Academy, Canberra ACT 2600, Australia.
         {\tt\small i.r.petersen@gmail.com} } 
}%
        
\def\begce{\begin{center}}
\def\endce{\end{center}}
\def\begar{\begin{array}}
\def\endar{\end{array}}
\def\begeq{\begin{equation}}
\def\endeq{\end{equation}}
\def\begdi{\begin{displaymath}}
\def\enddi{\end{displaymath}}
\def\begdis{\begin{eqnarray*}}
\def\enddis{\end{eqnarray*}}
\def\begeqa{\begin{eqnarray}}
\def\endeqa{\end{eqnarray}} 

\def\re{{\mathbb R}}
\def\C{{\mathbb C}} 
\begin{document}

\maketitle
\thispagestyle{empty}
\pagestyle{empty}

\begin{abstract}
This paper considers the problem of constructing a  direct coupled quantum observer network for a single qubit quantum system. The proposed  observer consists of a network of quantum harmonic oscillators and it is shown that the  observer network output converges to a consensus in a time averaged sense in which each component of the observer estimates a specified output of the quantum plant.  An example and simulations are included.
\end{abstract}

\section{Introduction} \label{sec:intro}
There has been significant interest in controlling  multi-agent systems to achieve a consensus; e.g., see \cite{LZHD14,SJ13}. Also, the problem of consensus in multi-agent estimation problems has been considered; e.g., see \cite{OS09}. In addition, consensus has been considered in  quantum multi-agent systems; see \cite{TMS14,SDPJ1a}. 
The papers \cite{PET14Aa,PET14Ba}  considered the problem of constructing a direct coupling quantum observer for a given quantum system. The problem of constructing an observer for  a linear quantum system has been considered for example in \cite{MJ12a}. The theory of linear quantum systems has been of considerable interest in recent years; e.g., see \cite{JNP1,ShP5}.
For such  system models, an important class of  control problems are  coherent
quantum feedback control problems; e.g., see \cite{JNP1,HM12}. In these  control problems, both the plant and the controller are quantum systems and the controller is designed to optimize some performance index. The coherent quantum observer problem can be regarded as a special case of the coherent
quantum feedback control problem in which the objective of the observer is to estimate the system variables of the quantum plant. The papers \cite{PET14Aa,PET14Ba} considered a direct coupling coherent observer problem in which the observer is directly coupled to the plant and not coupled via a field as in previous papers. This leads the papers \cite{PET14Aa,PET14Ba} to consider a notion of time-averaged convergence for the observers. 

We extend the results of \cite{PET14Ba} to consider a direct coupled  quantum observer for a single qubit quantum plant,  which is  a network of quantum harmonic oscillators.  This quantum network is constructed so that each output converges to the plant output of interest in a time averaged sense. This is a form of time averaged quantum consensus. 

\section{Quantum  Systems}
\noindent
{\bf Quantum Plant}

We first consider  the dynamics of a single qubit spin system which will correspond to the quantum plant; see also \cite{EMPUJ1a}.  
 The quantum mechanical behavior of the system is described in terms of the system \emph{observables} which are self-adjoint operators on the complex Hilbert space $\mathfrak{H}_p = \C^2$.   The commutator of two scalar operators $x$ and $y$ in ${\mathfrak{H}_p}$ is  defined as $[x, y] = xy - yx$.~Also, for a  vector of operators $x$ in ${\mathfrak H}_p$, the commutator of ${x}$ and a scalar operator $y$ in ${\mathfrak{H}_p}$ is the  vector of operators $[{x},y] = {x} y - y {x}$. 

The vector of system variables for the single qubit spin system under consideration is 
\begdi x_p=(x_1,x_2,x_3)^T\triangleq (\sigma_1,\sigma_2,\sigma_3),\enddi 
where $\sigma_1$, $\sigma_2$ and $\sigma_3$ are spin operators. Here, $x_p$ a  vector of self-adjoint operators, i.e., $x_p=x_p^\#$.~In particular $x_p(0)$ is represented by the Pauli matrices; i.e.,
\begin{eqnarray*}
\sigma_1(0)&=&\left(\begin{array}{cc}
         0 & 1 \\ 1 & 0
        \end{array} \right),\;\; 
\sigma_2(0)=\left(\begin{array}{cc}
         0 & -{\pmb i} \\ {\pmb i} & 0
        \end{array} \right),\\
\sigma_3(0)&=&\left(\begin{array}{cc}
         1 & 0 \\ 0 & -1
        \end{array} \right).
\end{eqnarray*}
The commutation relations for the spin operators are
\begeq \label{eq:Pauli_CCR}
[\sigma_i,\sigma_j] = 2{\pmb i} \sum_{k}\epsilon_{ijk}\sigma_k,
\endeq
where  $\epsilon_{ijk}$ denotes the Levi-Civita tensor. The dynamics of the system variables $x$ are determined by the system Hamiltonian which is a self-adjoint operator on $\mathfrak{H}_p$. The {Hamiltonian} is chosen to be linear in $x_p$; i.e., 
\begdi {\mathcal{H}}_p=r_p^T x_p(0) \;\;\enddi
where $r_p\in \re^3$. 
The plant model is then given by the differential equation
\begin{eqnarray}
\dot x_p(t) &=& -{\pmb i}[x_p(t),\mathcal{H}_p]; \nonumber \\
&=& A_px_p(t); \quad x_p(0)=x_{0p}; \nonumber \\
z_p(t) &=& C_px_p(t)
 \label{plant}
\end{eqnarray}
where $z_p$ denotes the  system variable to be estimated by the observer and  $C_p\in \rbb^{1\times 3}$; e.g., see \cite{EMPUJ1a}. Also, $A_p \in \re^{3\times 3}$. In order to obtain an expression for the matrix $A_p$ in terms of $r_p$, we  define the linear mapping
$\Theta_p: \C^3 \rightarrow \C^{3\times 3}$ as
\begeq \label{eq:Theta_definition}
\Theta_p(\beta)= \left(\begin{array}{ccc}
         0 & \beta_3 & -\beta_2 \\ -\beta_3 & 0 & \beta_1 \\ \beta_2  & -\beta_1  & 0
        \end{array} \right).
\endeq  
Then, it was shown in \cite{EMPUJ1a} that 
\begin{equation}
-{\pmb i}[x_p(t),r_p^T x_p(t)] = - 2  \Theta_p(r_p) x_p(t)
\label{Ap}
\end{equation}
and hence $A_p = - 2  \Theta_p(r_p)$. 

In addition, it is shown in \cite{EMPUJ1a} that the mapping $\Theta_p(\cdot)$ has the following properties:
\begin{eqnarray}
\label{eq:Theta_1} \Theta_p(\beta)\gamma &=& - \Theta_p(\gamma) \beta,\\
 \label{eq:Theta_beta_beta} \Theta_p(\beta)\beta &=& 0,\\
 \label{eq:Theta_multiplication} \Theta_p(\beta)\Theta_p(\gamma) &=& \gamma \beta^T -\beta^T \gamma I,\\
 \label{eq:Theta_composition} \Theta\left(\Theta_p(\beta)\gamma\right)&=&\Theta_p(\beta)\Theta_p(\gamma) - \Theta_p(\gamma)\Theta_p(\beta).
\end{eqnarray}

\noindent
{\bf Quantum Observer Network}
The quantum observer network will be a linear quantum system of the form 
\begin{eqnarray}
\dot x(t) &=& Ax(t); \quad x(0)=x_0
 \label{quantum_system}
\end{eqnarray}
where $A$ is a real matrix in $\rbb^{n
\times n}$, and $ x(t) = [\begin{array}{ccc} x_1(t) & \ldots &
x_n(t)
\end{array}]\trp$ is a vector of system observables which are self-adjoint operators on an infinite dimensional Hilbert space $\mathfrak{H}$; e.g., see \cite{JNP1}. Here $n$ is assumed to be an even number and $\frac{n}{2}$ is the number of modes in the quantum system. 

The initial system variables $x(0)=x_0$ 
are assumed to satisfy the {\em commutation relations}
\begin{equation}
[x_j(0), x_k(0) ] = 2 i \Theta_{jk}, \ \ j,k = 1, \ldots, n,
\label{x-ccr}
\end{equation}
where $\Theta_o$ is a real skew-symmetric matrix with components
$\Theta_{jk}$.  The matrix $\Theta_o$ is assumed to be  of the  form
\begin{equation}
\label{Theta}
\Theta_o=\diag(J,J,\ldots,J)
\end{equation}
 where $J$ denotes the real skew-symmetric $2\times 2$ matrix
$$
J= \left[ \begin{array}{cc} 0 & 1 \\ -1 & 0
\end{array} \right].$$

The system dynamics (\ref{quantum_system}) are determined by the system Hamiltonian
which is a self-adjoint operator on the underlying  Hilbert space  $\mathfrak{H}$. For the linear quantum systems under consideration, the system Hamiltonian will be a
quadratic form
$\mathcal{H}=\half x(0)\trp R x(0)$, where $R$ is a real symmetric matrix. Then, the corresponding matrix $A$ in 
(\ref{quantum_system}) is given by 
\begin{equation}
A=2\Theta_o R \label{eq_coef_cond_A}.
\end{equation}
 where $\Theta_o$ is defined as in (\ref{Theta}).
e.g., see \cite{JNP1}. In this case, the system is said to be {\em physically realizable} and the commutation relations hold for all times greater than zero:
\begin{eqnarray}
\label{CCR}
[x_o(t),x_o(t)^T]&=&x_o(t)x_o(t)^T- \left(x_o(t)x_o(t)^T\right)^T\nonumber \\
&=& 2{\pmb i}\Theta_o \ \mbox{for all } t\geq 0.
\end{eqnarray}

\begin{remark}
\label{R1}
Note that that the Hamiltonian $\mathcal{H}$ is preserved in time for the system (\ref{quantum_system}). Indeed,
$ \mathcal{\dot H} = \frac{1}{2}\dot{x}^TRx+\frac{1}{2}x^TR\dot{x} = -x^TR\Theta_o R x + x^TR\Theta_o R x = 0$ since $R$ is symmetric and $\Theta_o$ is skew-symmetric.
\end{remark}

We now describe the linear quantum system of the form (\ref{quantum_system}) which will correspond to the  quantum observer network; see also \cite{JNP1,GJ09,ZJ11}. 
This system is described by a non-commutative differential equation of the form
\begin{eqnarray}
\dot x_o(t) &=& A_ox_o(t);\quad x_o(0)=x_{0o};\nonumber \\
z_o(t) &=& C_ox_o(t)
 \label{observer}
\end{eqnarray}
where the observer output $z_o(t)$ is the  observer network estimate vector and $ A_p \in \rbb^{n_o
\times n_o}$, $C_o\in \rbb^{\frac{n_o}{2} \times n_o}$.  Also,  $x_o(t)$  is the vector of self-adjoint 
non-commutative system variables; e.g., see \cite{JNP1}. We assume the  observer network order $n_o$  is an even number with $N=\frac{n_o}{2}$ being the number of elements in the  quantum observer network. We also assume that the plant variables commute with the observer variables. The system dynamics (\ref{observer}) are determined by the observer system Hamiltonian  
which is a self-adjoint operator on the underlying  Hilbert space for the observer. For the  quantum observer network under consideration, this Hamiltonian is given by a 
quadratic form:
$\mathcal{H}_o=\half x_o(0)\trp R_o x_o(0)$, where $R_o$ is a real symmetric matrix. Then, the corresponding matrix $A_o$ in 
(\ref{observer}) is given by 
\begin{equation}
A_o=2\Theta_o R_o \label{eq_coef_cond_Ao}
\end{equation}
 where $\Theta_o$ is defined as in (\ref{Theta}). Furthermore, we will assume that the  quantum observer network has a graph structure and is coupled to the quantum plant as illustrated in Figure \ref{F1}. 
\begin{figure}[htbp]
\begin{center}
\includegraphics[width=8cm]{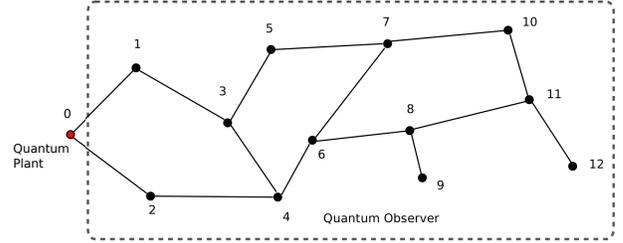}
\end{center}
\caption{The graph $(\mathcal{G},E)$ for a typical quantum observer network.}
\label{F1}
\end{figure}

The combined plant observer system is described by a connected graph $(\mathcal{G},E)$ which has $N+1$ nodes with node $0$ corresponding to the quantum plant and the remaining nodes, labelled $1,2,\ldots,N$, corresponding to the observer elements. 
This corresponds to an observer Hamiltonian of the form
\begin{eqnarray*}
\mathcal{H}_o&=&\half x_o(0)\trp R_o x_o(0) \nonumber \\
&=& \half\sum_{i=1}^{N}x_{oi}(0)\trp R_{oi} x_{oi}(0)\nonumber \\
&&+ \half\sum_{i=1}^{N}\sum_{j=1}^{N}x_{oi}(0)\trp R_{cij} x_{oj}(0)
\end{eqnarray*}
where the vector of observer system variables $x_o$ is partitioned according to each element of the  quantum observer network as follows
\[
x_o = \left[\begin{array}{l}x_{o1}\\x_{o2}\\\vdots\\x_{oN}\end{array}\right].
\]
We assume that the variables for each element of the  quantum observer network commute with the variables of all other elements of the  quantum observer network; i.e., 
\[
[x_{oi},x_{oj}^T]=0 ~~ \forall ~~i\neq j.
\]
Here, $x_{oi} = \left[\begin{array}{l}q_{oi}\\p_{oi}\end{array}\right]$ for $i=1,2,\ldots,N$ where
$q_{oi}$ is the  position operator for the $i$th observer element  and $p_{oi}$ is the  momentum operator for the $i$th observer element. 

In addition, we define a coupling Hamiltonian which defines the coupling between the quantum plant and the quantum observer network:
\[
\mathcal{H}_c = \sum_{i=1}^{N}x_{p}(0)\trp R_{c0i} x_{oi}(0).
\]
Furthermore, we write
\[
z_o = \left[\begin{array}{l}z_{o1}\\z_{o2}\\\vdots\\z_{oN}\end{array}\right]
\]
where 
\[
z_{oi} = C_{oi}x_{oi} \mbox{ for }i=1,2,\ldots,N.
\]
Then
\[
C_o=\left[\begin{array}{llll}C_{o1} & & &\\
                               & C_{o2} & 0 &\\
                               & 0 & \ddots & \\
                               &&& C_{oN}
\end{array}\right].
\]

Note that $R_{oi} \in \rbb^{2 \times 2}$, $R_{cij} \in \rbb^{2 \times 2}$, $C_{oi} \in \rbb^{1 \times 2}$, and each matrix $R_{oi}$ is symmetric for $i=1,2,\ldots,N$, $j=1,2,\ldots,N$. In addition, $R_{c0j} \in \rbb^{3 \times 2}$ for $j=1,2,\ldots,N$. Also, the matrices $R_{cij}$  for $i=0,1,\ldots,N$, $j=1,2,\ldots,N$ are such that $R_{cij} \neq 0$ if and only if $(i,j) \in E$, the set of edges for the graph $(\mathcal{G},E)$.

The augmented quantum linear system consisting of the quantum plant and the   quantum observer network is  described by the total Hamiltonian
\begin{eqnarray}
\mathcal{H}_a &=& \mathcal{H}_p+\mathcal{H}_c+\mathcal{H}_o\nonumber \\
&=& r_p^T x_p(0)+\half\sum_{i=1}^{N}x_{oi}(0)\trp R_{oi} x_{oi}(0)\nonumber \\
&&+\half\sum_{i=1}^{N}\sum_{j=1}^{N}x_{oi}(0)\trp R_{cij} x_{oj}(0)\nonumber \\
&&+\sum_{i=1}^{N}x_{p}(0)\trp R_{c0i} x_{oi}(0).
\label{total_hamiltonian}
\end{eqnarray}
Then, it follows that the augmented quantum  system is described by the equations
\begin{eqnarray}
\dot x_p(t) &=& -{\pmb i}[x_p(t),\mathcal{H}_a];~ x_p(0)=x_{0p};\nonumber \\
\dot x_o(t) &=& -{\pmb i}[x_o(t),\mathcal{H}_a];~ x_o(0)=x_{0o};\nonumber \\
z_p(t) &=& C_px_p(t);\nonumber \\
z_o(t) &=& C_ox_o(t);
\label{augmented_system}
\end{eqnarray}
e.g., see \cite{EMPUJ1a}.

We now formally define the notion of a direct coupled linear quantum observer network.

\begin{definition}
\label{D1}
The matrices $R_{oj}$, $R_{cij}$,  $C_{oj}$ for $i=0,1,\ldots,N$, $j=1,2,\ldots,N$ and the graph $(\mathcal{G},E)$ define a {\em  linear quantum observer network} achieving time-averaged consensus convergence for the single qubit quantum plant (\ref{plant}) if the corresponding augmented linear quantum system (\ref{augmented_system}) is such that
\begin{equation}
\label{average_convergence}
\lim_{T \rightarrow \infty} \frac{1}{T}\int_{0}^{T}(\left[\begin{array}{l}1\\1\\\vdots\\1\end{array}\right]z_p(t) 
- z_o(t))dt = 0.
\end{equation}
\end{definition}

\section{Constructing a Direct Coupling Coherent Quantum Observer Network}
We now describe the construction of a direct coupled linear quantum observer network.  In this section, we assume that  $A_p =0$ in (\ref{plant}). This corresponds to $r_p = 0$ in the plant Hamiltonian. It follows from (\ref{plant}) that the vector of plant system variables $x_p(t)$ will remain fixed if the plant is not coupled to the observer network. However, when the plant is coupled to the quantum observer network this will no longer be the case. We will show that if the quantum observer is suitably designed, the plant quantity to be estimated  $z_p(t)$ will remain fixed and the condition (\ref{average_convergence}) will be satisfied. 

We assume that the matrices $R_{cij}$, $R_{oi}$ for $i=0,1,\ldots,N$, $j=1,2,\ldots,N$  are of the form  
\begin{equation}
\label{RciRoi}
R_{cij} = \alpha_{ij}\beta_{ij}^T,~~R_{oi} = \omega_iI
\end{equation}
where $\alpha_{ij}  \in \rbb^{2}$, $\beta_{ij} \in \rbb^{2}$ and $\omega_i > 0$  for $i=1,2,\ldots,N$, $j=1,2,\ldots,N$. Also, we assume that 
\begin{equation}
\label{Rc0j}
R_{c0j} = \alpha_{0j}\beta_{0j}^T \mbox{ where }\alpha_{0j} =  \alpha_{0}= C_p^T\in \rbb^{3}
\end{equation}
 for $j=1,2,\ldots,N$ such that $(0,j) \in E$, the set of edges for the graph $(\mathcal{G},E)$. In addition, note that $\alpha_{ij} = 0$ and $\beta_{ij} = 0$ for $(i,j) \not\in E$. Furthermore, we assume 
\begin{equation}
\label{Coi}
C_{oi} = C_p= \alpha_0^T
\end{equation}
 for $i=1,2,\ldots,N$.

We will show that these assumptions imply that the quantity $z_p(t) = C_px_p(t)$ will be constant for the augmented quantum system (\ref{augmented_system}). Indeed, the total Hamiltonian (\ref{total_hamiltonian}) will be given by 
\begin{eqnarray*}
\mathcal{H}_a &=& \half\sum_{i=1}^{N}\omega_ix_{oi}(0)\trp  x_{oi}(0)\nonumber \\
&&+\half\sum_{i=1}^{N}\sum_{j=1}^{N}x_{oi}(0)\trp \alpha_{ij}\beta_{ij}^T x_{oj}(0)\nonumber \\
&&+\sum_{j=1}^{N}x_{p}(0)\trp\alpha_{0j}\beta_{0j}^T  x_{oj}(0).
\end{eqnarray*}

Now using a similar calculation as in (\ref{Ap}), we calculate 
\begin{eqnarray}
\label{xpt}
\dot x_p(t) &=& -{\pmb i}[x_p(t),\mathcal{H}_a]\nonumber \\
&=&-2\sum_{j=1}^{N}\Theta_p(\alpha_{0j})x_p(t)\beta_{0j}^Tx_o(t) \nonumber \\
&=& -2 \Theta_p(\alpha_{0})x_p(t)\sum_{(0,j) \in E} \beta_{0j}^Tx_o(t).
\end{eqnarray}
Hence, the quantity $z_p(t) = C_px_p(t)$  satisfies the differential equation
\begin{eqnarray}
\label{zop}
\dot z_p(t) &=& -2C_p\Theta_p(\alpha_{0})x_p(t)\sum_{(0,j) \in E} \beta_{0j}^Tx_o(t)\nonumber \\
 &=& -2\alpha^T_0\Theta_p(\alpha_{0})x_p(t)\sum_{(0,j) \in E} \beta_{0j}^Tx_o(t) \nonumber \\
&=& 0
\end{eqnarray}
using (\ref{eq:Theta_beta_beta}) and the fact that $\Theta_p(\alpha)$ is skew symmetric.  That is, the quantity $z_p(t)$ remains constant and is not affected by the coupling to the coherent quantum observer network:
\begin{equation}
\label{zp_const}
z_p(t) = z_p = z_p(0)~ \forall t \geq 0.
\end{equation}

Also to calculate $\dot x_o(t)$, we first observe that for any $i=0,1,\ldots,N$, $j=1,2,\ldots,N$.
\begin{eqnarray*}
\left[\beta_{ij}^Tx_{oj}(t),x_{oj}(t)\right] &=& \beta_{ij}^Tx_{oj}(t)x_{oj}(t)-x_{oj}(t)\beta_{ij}^Tx_{oj}(t) \nonumber  \\
&=& \left(\beta_{ij}^Tx_{oj}(t)x_{oj}(t)^T\right)^T\nonumber \\
&&-x_{oj}(t)x_{oj}(t)^T\beta_{ij}\nonumber\\
&=&\left(x_{oj}(t)x_{oj}(t)^T\right)^T\beta_{ij}\nonumber \\
&&-x_{oj}(t)x_{oj}(t)^T\beta_{ij}\nonumber\\
&=&-\left[x_{oj}(t),x_{oj}(t)^T\right]\beta_{ij}\nonumber\\
&=& -2 {\pmb i} J \beta_{ij} 
\end{eqnarray*}
using (\ref{CCR}). Hence, using this result and a similar approach to the derivation of (\ref{eq_coef_cond_A}) in \cite{JNP1}, we obtain 
\begin{eqnarray}
\label{xot}
\dot x_{oj}(t) &=& {\pmb i}[\mathcal{H}_a,x_{oj}(t)]\nonumber \\
&=&2\omega_jJx_{oj}(t)\nonumber \\
&&+\half{\pmb i}\sum_{i=1}^{N}\left(-2 {\pmb i} J \beta_{ij} \right)\alpha_{ij}^Tx_{oi}(t)\nonumber \\
&&+\half{\pmb i}\sum_{i=1}^{N}\left(-2 {\pmb i} J \alpha_{ji} \right)\beta_{ji}^Tx_{oi}(t)\nonumber \\
&&+{\pmb i}\alpha_{0j}\trp x_{p}(t)\left(-2 {\pmb i} J \beta_{0j} \right) \nonumber \\
&=&2\omega_jJx_{oj}(t)+J\sum_{i=1}^{N} \beta_{ij}\alpha_{ij}^Tx_{oi}(t)\nonumber \\
&&+J\sum_{i=1}^{N} \alpha_{ji}\beta_{ji}^Tx_{oi}(t)+2J\beta_{0j}z_p
\end{eqnarray}
for $j=1,2,\ldots,N$.

To construct a suitable  quantum observer network, we will further assume that 
\begin{equation}
\label{alphabeta}
\alpha_{ij} = \alpha_1, ~~\beta_{ij} = -\mu_{ij} \alpha_1
\end{equation}
for $i=1,\ldots,N$, $j=1,2,\ldots,N$ where $(i,j) \in E$. 
Here,  $\alpha_1 \in \rbb^2$ and
\begin{equation}
\label{muij}
\mu_{ij}=\mu_{ji} > 0. 
\end{equation}
Also, we will assume that 
\begin{equation}
\label{beta0}
\beta_{0j} = -\mu_{0j} \alpha_1
\end{equation}
for $j=1,2,\ldots,N$ where $(0,j) \in E$. 

In order to construct suitable values for the quantities $\mu_{ij}$ and $\omega_i$ so that (\ref{average_convergence}) is satisfied, we will require that 
\begin{eqnarray}
\label{xtilde}
&&2\omega_jJ\alpha_1-\sum_{(i,j)\in E,i> 0} \mu_{ij} J\alpha_1\alpha_1^T\alpha_1\nonumber \\
&&-\sum_{(i,j)\in E,i> 0}\mu_{ij}J\alpha_1 \alpha_1^T\alpha_1+2J\beta_{0j}\alpha_1^T\alpha_1 = 0
\end{eqnarray}
for $j=1,2,\ldots,N$. This condition is equivalent to 
\begin{eqnarray}
\label{mui1}
\omega_j&=&\sum_{(i,j)\in E,i> 0} \mu_{ij}\|\alpha_1\|^2 +\mu_{0j}\|\alpha_1\|^2
\end{eqnarray}
for $(0,j) \in E$ and 
\begin{eqnarray}
\label{mui2}
\omega_j=\sum_{(i,j)\in E,i> 0} \mu_{ij}\|\alpha_1\|^2
\end{eqnarray}
for  $(0,j) \not\in E$. 

Then, we define
\[
\tilde x_{oj}(t) = x_{oj}(t) - \frac{1}{\|\alpha_1\|^2}\alpha_1z_p
\]
for $j=1,2,\ldots,N$. It follows from (\ref{xtilde}) and (\ref{xot}) that
\begin{eqnarray*}
\dot{\tilde x}_{oj}(t)  &=& 2\omega_jJ\tilde x_{oj}(t)+J\sum_{i=1}^{N} \beta_{ij}\alpha_{ij}^T\tilde x_{oi}(t)\nonumber \\
&&+J\sum_{i=1}^{N} \alpha_{ji}\beta_{ji}^T\tilde x_{oi}(t)\nonumber \\
&=& 2\omega_jJ\tilde x_{oj}(t)-2\sum_{(i,j)\in E,i> 0} \mu_{ij} J\alpha_1\alpha_1^T\tilde x_{oi}(t)
\end{eqnarray*}
for $j=1,2,\ldots,N$.

We now write this equation as
\begin{equation}
\label{xtildedot}
\left[\begin{array}{l}
\dot{\tilde x}_{o1}(t)\\
\dot{\tilde x}_{o2}(t)\\
\vdots\\
\dot{\tilde x}_{oN}(t)
\end{array}\right] 
= A_o \left[\begin{array}{l}
\tilde x_{o1}(t)\\
\tilde x_{o2}(t)\\
\vdots\\
\tilde x_{oN}(t)
\end{array}\right] 
\end{equation}
where $A_o$ is an $N\times N$ block matrix with blocks
\[
a_{oij} = \left\{\begin{array}{ll}2\omega_iJ & \mbox{ for }i=j,\\
-2\mu_{ij} J\alpha_1\alpha_1^T & \mbox{ for }i\neq j \mbox{ and }(i,j)\in E,\\
0 & \mbox{ otherwise}
\end{array}\right.
\]
for $i=1,2,\ldots,N$, $j=1,2,\ldots,N$. Also, $A_o$ is as given in  (\ref{eq_coef_cond_Ao}) 
where $R_o$ is a symmetric $N\times N$ block matrix with blocks
\[
r_{oij} = \left\{\begin{array}{ll}\omega_iI & \mbox{ for }i=j,\\
-\mu_{ij} \alpha_1\alpha_1^T & \mbox{ for }i\neq j \mbox{ and }(i,j)\in E,\\
0 & \mbox{ otherwise}
\end{array}\right.
\]
for $i=1,2,\ldots,N$, $j=1,2,\ldots,N$.

To show that the above candidate  quantum observer network leads to the satisfaction of the condition (\ref{average_convergence}), we  note that 
\[
\tilde x_o =  \left[\begin{array}{l}
\tilde x_{o1}\\
\tilde x_{o2}\\
\vdots\\
\tilde x_{oN}
\end{array}\right] 
\]
 satisfies (\ref{xtildedot}). Hence, if we can show that 
\begin{equation}
\label{xtildeav}
\lim_{T \rightarrow \infty} \frac{1}{T}\int_{0}^{T}\tilde x_o(t)dt = 0
\end{equation}
 then it will follow from 
\begin{eqnarray}
\label{Coxtilde}
\lefteqn{C_o\frac{1}{\|\alpha_1\|^2}\left[\begin{array}{l}\alpha_1\\\alpha_1\\\vdots\\\alpha_1\end{array}\right]z_p}\nonumber \\
&=& \frac{1}{\|\alpha_1\|^2}\left[\begin{array}{llll}\alpha_1^T & & &\\
                               & \alpha_1^T & 0 &\\
                               & 0 & \ddots & \\
                               &&& \alpha_1^T
\end{array}\right]\left[\begin{array}{l}\alpha_1\\\alpha_1\\\vdots\\\alpha_1\end{array}\right]z_p \nonumber \\
&=& \left[\begin{array}{l}1\\1\\\vdots\\1\end{array}\right]z_p
\end{eqnarray}
 that (\ref{average_convergence}) is satisfied.

We now show that the symmetric matrix $R_o$ is positive-definite.

\begin{lemma}
\label{L1}
The matrix $R_o$ is positive definite.
\end{lemma}

\begin{proof}
In order to establish this lemma, let 
\[
x_o =  \left[\begin{array}{l}x_{o1}\\x_{o2}\\\vdots\\x_{oN}\end{array}\right]
\]
be a non-zero real vector. Then
\begin{eqnarray}
\label{Roineq}
x_o^TR_ox_o &=& \sum_{i=1}^N \omega_i\|x_{oi}\|^2\nonumber \\
&&-\sum_{(i,j)\in E,i> 0,j>0}\mu_{ij}x_{oi}^T\alpha_1 x_{oj}^T\alpha_1\nonumber \\
&\geq& \sum_{i=1}^N \omega_i\|x_{oi}\|^2\nonumber \\
&&-\sum_{(i,j)\in E,i> 0,j>0}\mu_{ij}\|x_{oi}\|\|x_{oj}\|\|\alpha_1\|^2\nonumber \\
&=&\sum_{i=1}^N \omega_i\|x_{oi}\|^2\nonumber \\
&&-\sum_{(i,j)\in E,i> 0,j>0}\tilde \mu_{ij}\|x_{oi}\|\|x_{oj}\|\nonumber \\
\end{eqnarray}
using the Cauchy-Schwarz inequality. Here,
\[
\tilde \mu_{ij} = \mu_{ij}\|\alpha_1\|^2
\]
for $0=1,2,\ldots,N$, $j=1,2,\ldots,N$.
Thus, (\ref{Roineq}) implies
\[
x_o^TR_ox_o \geq \check x_o^T\tilde R_o\check x_o
\]
where
\[
\check x_o = \left[\begin{array}{c}\|x_{o1}\|\\\|x_{o2}\|\\\vdots\\\|x_{oN}\|\end{array}\right]
\]
and $\tilde R_o$ is a symmetric $N\times N$ matrix with elements defined by
\[
\tilde r_{oij} = \left\{\begin{array}{ll}\omega_i & \mbox{ for }i=j,\\
-\tilde \mu_{ij} & \mbox{ for }i\neq j \mbox{ and }(i,j)\in E,\\
0 & \mbox{ otherwise}
\end{array}\right.
\]
for $i=1,2,\ldots,N$, $j=1,2,\ldots,N$.

Now the vector $\check x_o$ will be non-zero if and only if the vector $x_o$ is non-zero. Hence, the matrix $R_o$ will be positive-definite if we can show that the matrix $\tilde R_o$ is positive-definite. In order to establish this fact, we first note that (\ref{mui1}) and (\ref{mui2})  imply that
\begin{eqnarray*}
\omega_j&=&\sum_{(i,j)\in E,i> 0} \tilde \mu_{ij} +\tilde \mu_{0j}
\end{eqnarray*}
for $(0,j) \in E$ and 
\begin{eqnarray*}
\omega_j=\sum_{(i,j)\in E,i> 0} \mu_{ij}\|\alpha_1\|^2
\end{eqnarray*}
for  $(0,j) \not\in E$. 
Hence, we can write
\begin{eqnarray*}
\tilde R_o 
&=& \tilde R_{o1} + \tilde R_{o2}
\end{eqnarray*}
where $\tilde R_{o1}$ is a symmetric $N\times N$ matrix with elements defined by
\[
\tilde r_{o1ij} = \left\{\begin{array}{ll}\sum_{(k,j)\in E,k> 0} \tilde \mu_{kj} & \mbox{ for }i=j,\\
-\tilde \mu_{ij} & \mbox{ for }i\neq j \mbox{ and }(i,j)\in E,\\
0 & \mbox{ otherwise}
\end{array}\right.
\]
for $i=1,2,\ldots,N$, $j=1,2,\ldots,N$. Also, $\tilde R_{o2}$ is a diagonal $N\times N$ matrix with elements defined by
\[
\tilde r_{o2ij} = \left\{\begin{array}{ll}\tilde \mu_{0j} & \mbox{ for }i=j \mbox{ and } (0,j) \in E,\\
0 & \mbox{ otherwise}
\end{array}\right.
\]
It follows that the matrix $\tilde R_{o2}$  is positive semidefinite. 

Now the matrix $\tilde R_{o1}$ is the Laplacian matrix for the weighted graph $(\mathcal{\tilde G}, \tilde E)$ obtained by removing node $0$ from the graph $(\mathcal{G}, E)$ along with the associated edges. Then each edge $(i,j) \in \tilde E$ is given a weight 
$\tilde \mu_{ij}$; e.g., see Figure \ref{F2} which shows the weighted graph $(\mathcal{\tilde G}, \tilde E)$ which would correspond to the graph $(\mathcal{G}, E)$ shown in Figure \ref{F1}. 
\begin{figure}[htbp]
\begin{center}
\includegraphics[width=8cm]{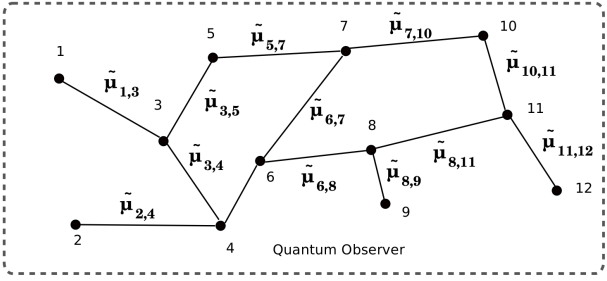}
\end{center}
\caption{The weighted graph $(\mathcal{\tilde G}, \tilde E)$ corresponding to the graph $(\mathcal{G}, E)$ in Figure \ref{F1}.}
\label{F2}
\end{figure}

It follows that the matrix $\tilde R_{o1}$ is positive-semidefinite with null space of the following form:
\[
\mathcal{N}(\tilde R_{o1}) = \mbox{span}\{f_1,f_2,\ldots,f_m\}
\]
where $m$ is the number of connected components of the graph $(\mathcal{\tilde G}, \tilde E)$. Also, each of the vectors $f_1,f_2,\ldots,f_m$ are vectors whose elements are either zeros or ones. For the vector $f_k$, the elements of this vector which are ones correspond to  the nodes in the graph $(\mathcal{\tilde G}, \tilde E)$ in the $k$th connected component. 

The fact that $\tilde R_{o1} \geq 0$ and $\tilde R_{o2} \geq 0$ implies that $\tilde R_{o} \geq 0$. In order to show that $\tilde R_{o} > 0$, suppose that $x$ is a non-zero vector in $\mathcal{N}(\tilde R_{o})$. It follows that 
\[
x^T\tilde R_{o}x = x^T\tilde R_{o1}x+x^T\tilde R_{o2}x = 0.
\]
Since $\tilde R_{o1} \geq 0$ and $\tilde R_{o2} \geq 0$, $x$ must be contained in the null space of $\tilde R_{o1}$ and the null space of $\tilde R_{o2}$. Therefore $x$ must be of the form
\[
x = \sum_{k=1}^m\gamma_kf_k
\]
where not all $\gamma_k = 0$. However, since the graph $(\mathcal{G}, E)$ is connected, it follows that there must be at least one branch $(0,j) \in E$ to a node in each of the connected components in the graph $(\mathcal{\tilde G}, \tilde E)$. Then
\[
x^T\tilde R_{o2}x = \sum_{(0,j) \in E}\tilde \mu_{0,j}\gamma_{k(j)}^2 = 0
\]
where $k(j)$ corresponds to the node of the connected component in $(\mathcal{\tilde G}, \tilde E)$ which the branch $(0,j)$ connects to. Since each $\tilde \mu_{0,j} > 0$, it follows that 
\[
\gamma_{k(j)} = 0
\]
for all $(0,j) \in E$. Furthermore, since each connected component in $(\mathcal{\tilde G}, \tilde E)$ has at least one branch $(0,j) \in E$ connected to it, it follows that $\gamma_1 = \gamma_2 \ldots = \gamma_m = 0$. However, this contradicts the assumption that not all $\gamma_k = 0$.  Thus, we can conclude that the matrix $\tilde R_{o}$ is positive definite and hence, the matrix  $R_{o}$ is positive definite. This completes the proof of the lemma. 
\end{proof}

We now verify that the condition (\ref{average_convergence}) is satisfied for the   quantum observer network under consideration. We recall from Remark \ref{R1} that the quantity $\half \tilde x_o(t)\trp R_o \tilde x_o(t)$
remains constant in time for the linear system:
\[
\dot{\tilde x}_o = A_o\tilde x_o= 2\Theta R_o \tilde x_o.
\]
That is 
\begin{equation}
\label{Roconst}
\half \tilde x_o(t) \trp R_o \tilde x_o(t) = \half \tilde x_o(0) \trp R_o \tilde x_o(0) \quad \forall t \geq 0.
\end{equation}
However, $\tilde x_o(t) = e^{2\Theta R_ot}\tilde x_o(0)$ and $R_o > 0$. Therefore, it follows from (\ref{Roconst}) that
\[
\sqrt{\lambda_{min}(R_o)}\|e^{2\Theta R_ot}\tilde x_o(0)\| \leq \sqrt{\lambda_{max}(R_o)}\|\tilde x_o(0)\|
\]
for all $\tilde x_o(0)$ and $t \geq 0$. Hence, 
\begin{equation}
\label{exp_bound}
\|e^{2\Theta R_ot}\| \leq \sqrt{\frac{\lambda_{max}(R_o)}{\lambda_{min}(R_o)}}
\end{equation}
for all $t \geq 0$.

Now since $\Theta $ and $R_o$ are non-singular,
\[
\int_0^Te^{2\Theta R_ot}dt = \half e^{2\Theta R_oT}R_o^{-1}\Theta ^{-1} - \half R_o^{-1}\Theta ^{-1}
\]
and therefore, it follows from (\ref{exp_bound}) that
\begin{eqnarray*}
\lefteqn{\frac{1}{T} \|\int_0^Te^{2\Theta R_ot}dt\|}\nonumber \\
 &=& \frac{1}{T} \|\frac{1}{2}e^{2\Theta R_oT}R_o^{-1}\Theta ^{-1} - \frac{1}{2}R_o^{-1}\Theta ^{-1}\|\nonumber \\
&\leq& \frac{1}{2T}\|e^{2\Theta R_oT}\|\|R_o^{-1}\Theta ^{-1}\| \nonumber \\
&&+ \frac{1}{2T}\|R_o^{-1}\Theta ^{-1}\|\nonumber \\
&\leq&\frac{1}{2T}\sqrt{\frac{\lambda_{max}(R_o)}{\lambda_{min}(R_o)}}\|R_o^{-1}\Theta ^{-1}\|\nonumber \\
&&+\frac{1}{2T}\|R_o^{-1}\Theta ^{-1}\|\nonumber \\
&\rightarrow & 0 
\end{eqnarray*}
as $T \rightarrow \infty$. Hence,  
\begin{eqnarray*}
\lefteqn{\lim_{T \rightarrow \infty} \frac{1}{T}\|\int_{0}^{T} \tilde x_o(t)dt\| }\nonumber \\
&=& \lim_{T \rightarrow \infty}\frac{1}{T}\|\int_{0}^{T} e^{2\Theta R_ot}\tilde x_o(0)dt\| \nonumber \\
&\leq& \lim_{T \rightarrow \infty}\frac{1}{T} \|\int_{0}^{T} e^{2\Theta R_ot}dt\|\|\tilde x_o(0)\|\nonumber \\
&=& 0.
\end{eqnarray*}
This implies
\[
\lim_{T \rightarrow \infty} \frac{1}{T}\int_{0}^{T} \tilde x_o(t)dt = 0
\]
and hence, it follows from (\ref{xtildedot}) and (\ref{Coxtilde}) that
\[
\lim_{T \rightarrow \infty} \frac{1}{T}\int_{0}^{T} z_o(t)dt = \left[\begin{array}{l}1\\1\\\vdots\\1\end{array}\right]z_p.
\]

Also, (\ref{zp_const}) implies 
\[
\lim_{T \rightarrow \infty} \frac{1}{T}\int_{0}^{T} z_p(t)dt = \left[\begin{array}{l}1\\1\\\vdots\\1\end{array}\right]z_p.
\]
Therefore, condition (\ref{average_convergence}) is satisfied. Thus, we have established the following theorem.

\begin{theorem}
\label{T1}
Consider a single qubit quantum plant of the form (\ref{plant}) where $r_p = 0$ and hence  $A_p = 0$. Then the matrices $R_{oi} $, $R_{cij}$, $C_{oi}$, $R_{oi}$ for $i=1,2,\ldots,N$, $j=1,2,\ldots,N$ and the connected graph $(\mathcal{G}, E)$
  will define a  direct coupled quantum observer network achieving time-averaged consensus convergence for this quantum plant if the conditions (\ref{RciRoi}), (\ref{Rc0j}), (\ref{Coi}), (\ref{alphabeta}), (\ref{beta0}), (\ref{muij}), (\ref{mui1}), (\ref{mui2}) are satisfied. 
\end{theorem}

\begin{remark}
\label{R2}
The  quantum observer network constructed above is determined by the choice of the positive parameters $\mu_{ij}$ for $i=0,1,\ldots,N$, $j=1,2,\ldots,N$. A number of possible choices for these parameters could be considered. One choice is to choose all of these parameters to be the same as $\tilde \mu_{ij} = \omega_0$ for $i=0,1,\ldots,N$, $j=1,2,\ldots,N$ where $\omega_0 > 0$ is a frequency parameter. 
Another possible approach is to choose the parameters $\mu_{ij}$ for $i=0,1,\ldots,N$, $j=1,2,\ldots,N$  randomly with a uniform distribution on a suitable frequency interval. 
\end{remark}
\section{Illustrative Example}
We now present some numerical simulations to illustrate the direct coupled  quantum observer network described in the previous section. We choose the quantum plant to have $A_p = 0$ and $C_p = [1~~0~~0]$. That is, the variable to be estimated by the quantum observer is the spin operator $\sigma_1$ of the quantum plant. For the  quantum observer network, we choose $N=5$ so that the  quantum observer network has five elements. Also, we suppose that the graph $(\mathcal{G}, E)$ defining the plant observer network is the complete graph corresponding to the five observer nodes and the plant node; i.e., every node is connected to every other node in this graph. This graph is illustrated in Figure \ref{F3}.  In addition, we choose $\alpha_1= [1~~0]^T$ and as discussed in Remark \ref{R2}, we choose the parameters $\tilde \mu_{ij}$  so that $\tilde \mu_{ij} = \omega_0$ for $i=0,1,\ldots,N$, $j=1,2,\ldots,N$ where $\omega_0 =1$. Then the dynamics of the corresponding  quantum observer network are defined by equations (\ref{zop}) and (\ref{xot}). 

\begin{figure}[htbp]
\begin{center}
\includegraphics[width=8cm]{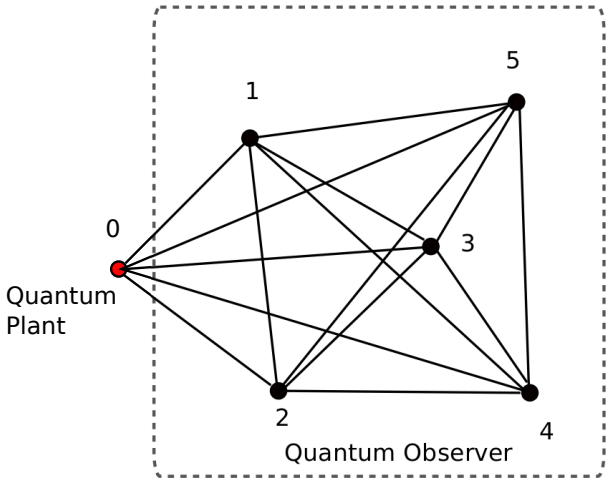}
\end{center}
\caption{The  plant observer network considered in the example.}
\label{F3}
\end{figure}

For this example, the augmented plant-observer system  can be described by the equations
\[
\dot x_a(t)= A_a x_a(t), \mbox{ where }
x_a(t) = \left[\begin{array}{l}  z_p(t)\\  x_{o1}(t)\\ x_{o2}(t)\\\vdots\\ x_{o5}(t)\end{array}\right]
\]
and
{\tiny
\[
A_a 
= \left[\begin{array}{lllllllllll}      
     0&     0&     0&     0&     0&     0&     0&     0&     0&     0&     0\\
     0&     0&    10&     0&     0&     0&     0&     0&     0&     0&     0\\
     2&   -10&     0&     2&     0&     2&     0&     2&     0&     2&     0\\
     0&     0&     0&     0&    10&     0&     0&     0&     0&     0&     0\\
     2&     2&     0&   -10&     0&     2&     0&     2&     0&     2&     0\\
     0&     0&     0&     0&     0&     0&    10&     0&     0&     0&     0\\
     2&     2&     0&     2&     0&   -10&     0&     2&     0&     2&     0\\
     0&     0&     0&     0&     0&     0&     0&     0&    10&     0&     0\\
     2&     2&     0&     2&     0&     2&     0&   -10&     0&     2&     0\\
     0&     0&     0&     0&     0&     0&     0&     0&     0&     0&    10\\
     2&     2&     0&     2&     0&     2&     0&     2&     0&   -10&     0 
\end{array}\right].
\]
}

Then, we can write
\[
x_a(t) = 
\Phi(t) x_a(0)
\]
where 
\[
\Phi(t)
= e^{A_a t}.
\]
Thus, the plant variable to be estimated $z_p(t)$ is given by
\begin{eqnarray*}
z_p(t) &=& e_1C_a\Phi(t) x_a(0) \nonumber \\
&=& \sum_{i=1}^{2N+2}e_1C_a\Phi_i(t) x_{ai}(0)
\end{eqnarray*}
where 
\[
C_a = \left[\begin{array}{ll}C_p & 0 \\0 & C_o\end{array}\right],
\]
$e_1$ is the first unit vector in the standard basis for $\rbb^{N+1}$, $\Phi_i(t)$ is the $i$th column of the matrix $\Phi(t)$ and
$x_{ai}(0)$ is the $i$th component of the vector $x_a(0)$.  We plot each of the quantities  
$e_1C_a\Phi_1(t),e_1C_a\Phi_2(t),\ldots,e_1C_a\Phi_{2N+1}(t) $ in Figure \ref{F5}(a). 
\begin{figure}%
\centering
\subfloat[][]{\includegraphics[width=3.9cm]{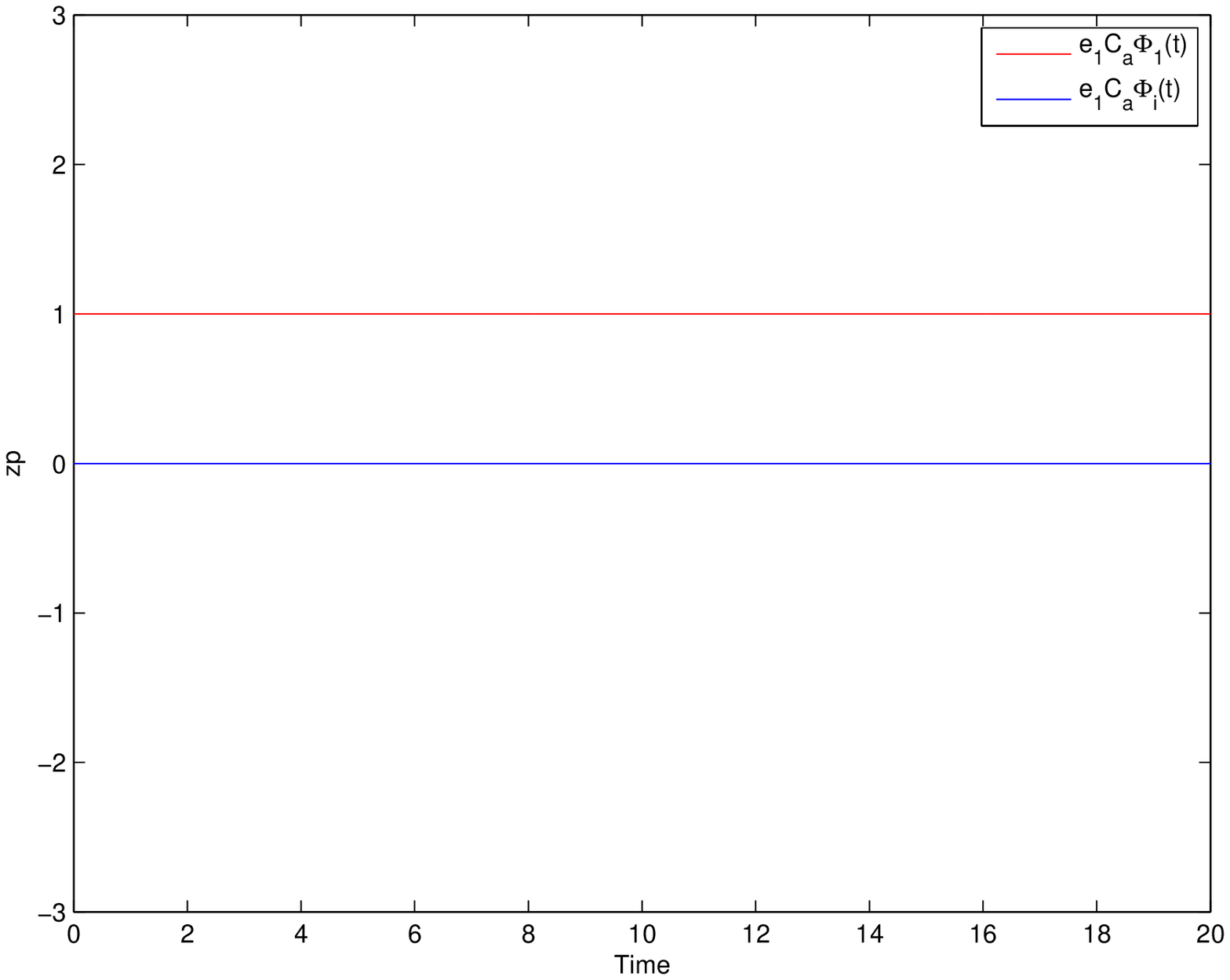}}%
\qquad
\subfloat[][]{\includegraphics[width=3.9cm]{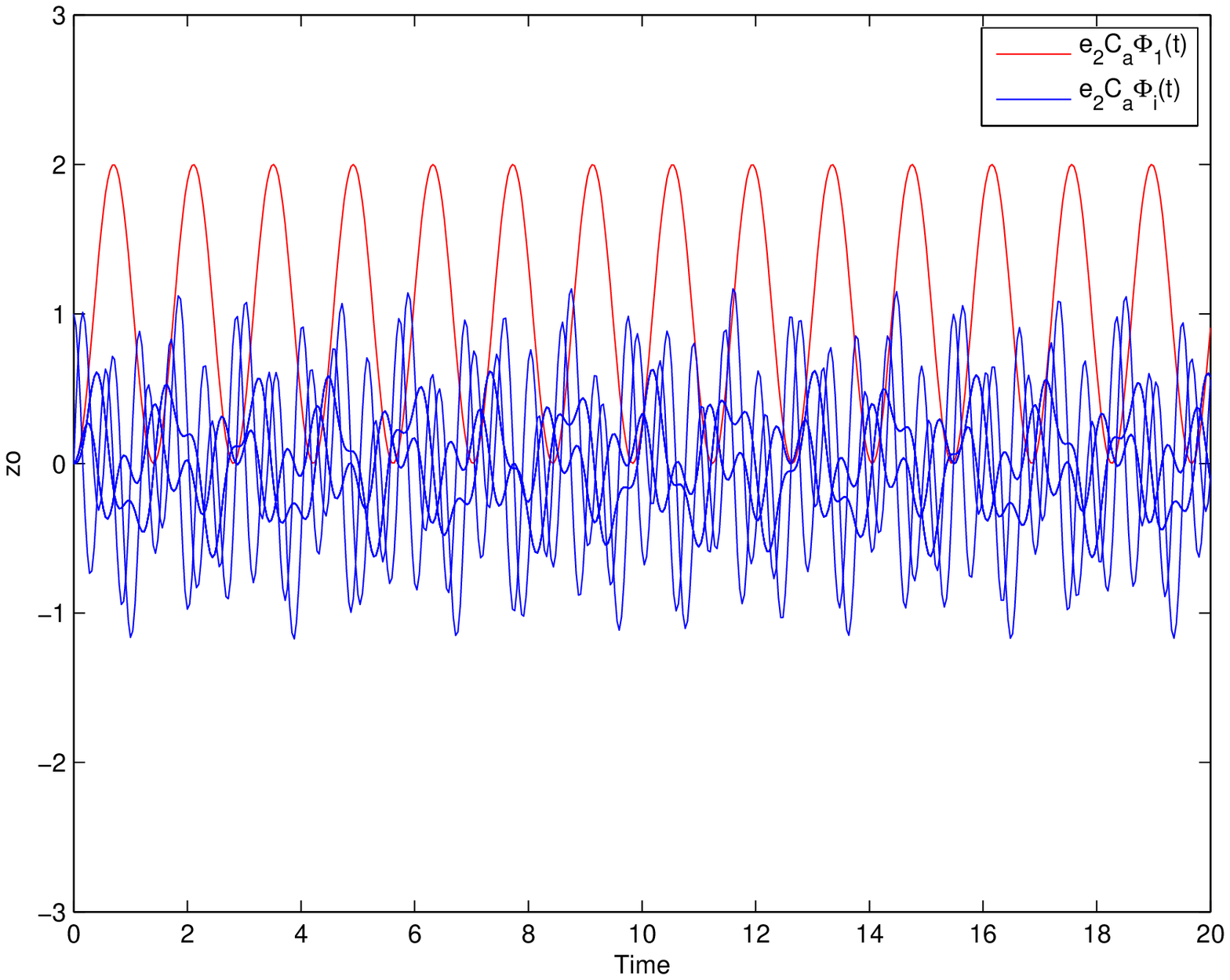}}\\
\subfloat[][]{\includegraphics[width=3.9cm]{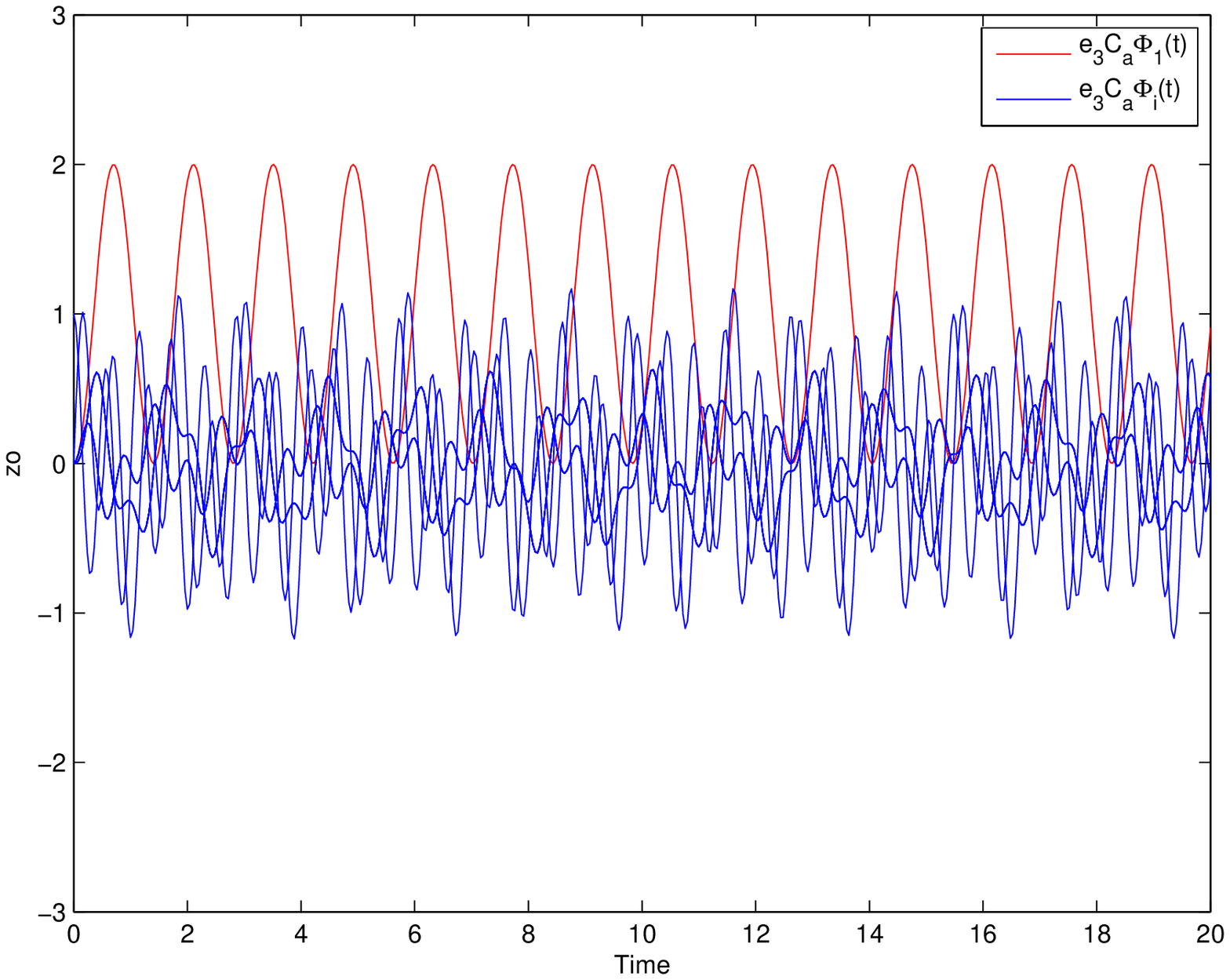}}%
\qquad
\subfloat[][]{\includegraphics[width=3.9cm]{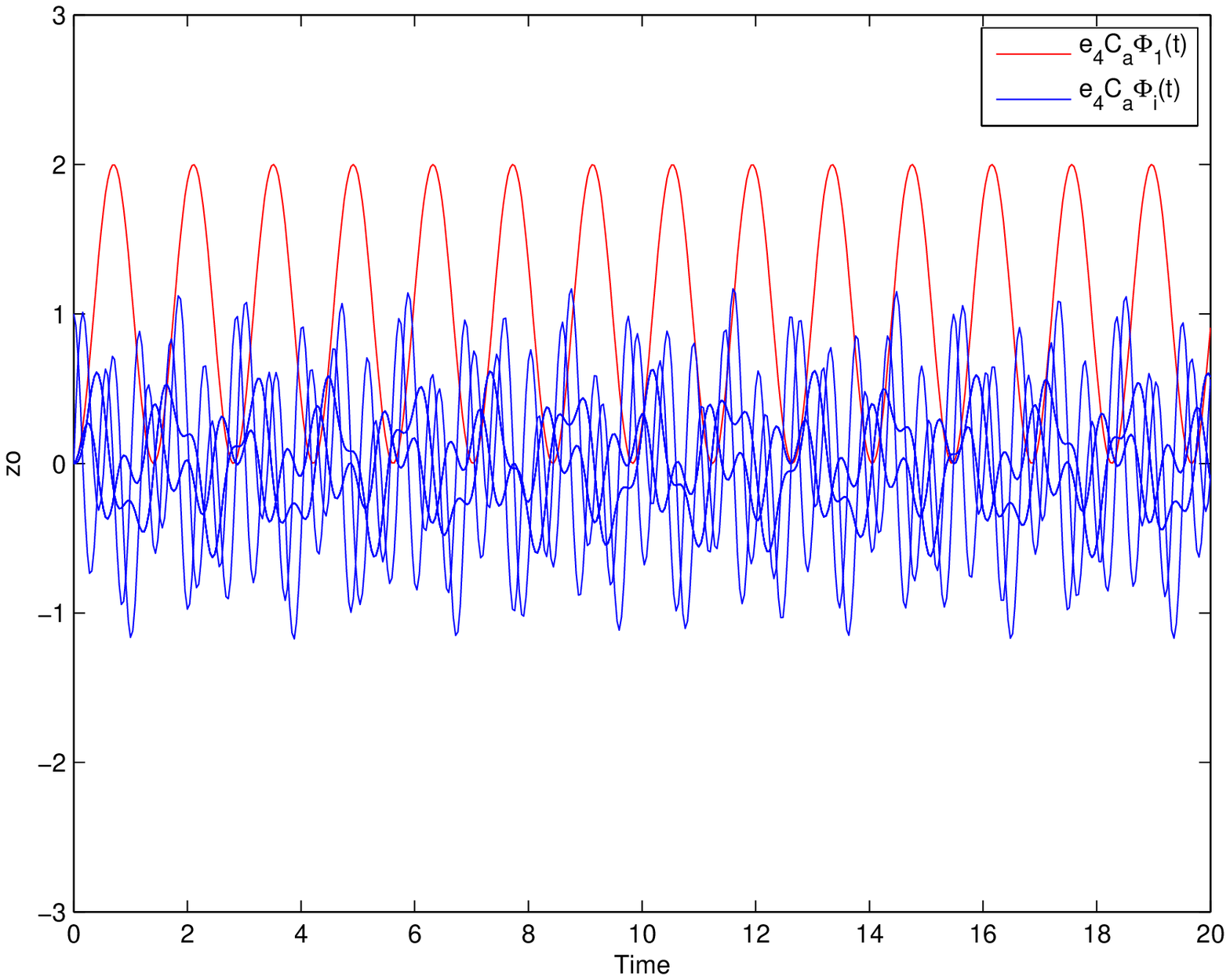}}\\
\subfloat[][]{\includegraphics[width=3.9cm]{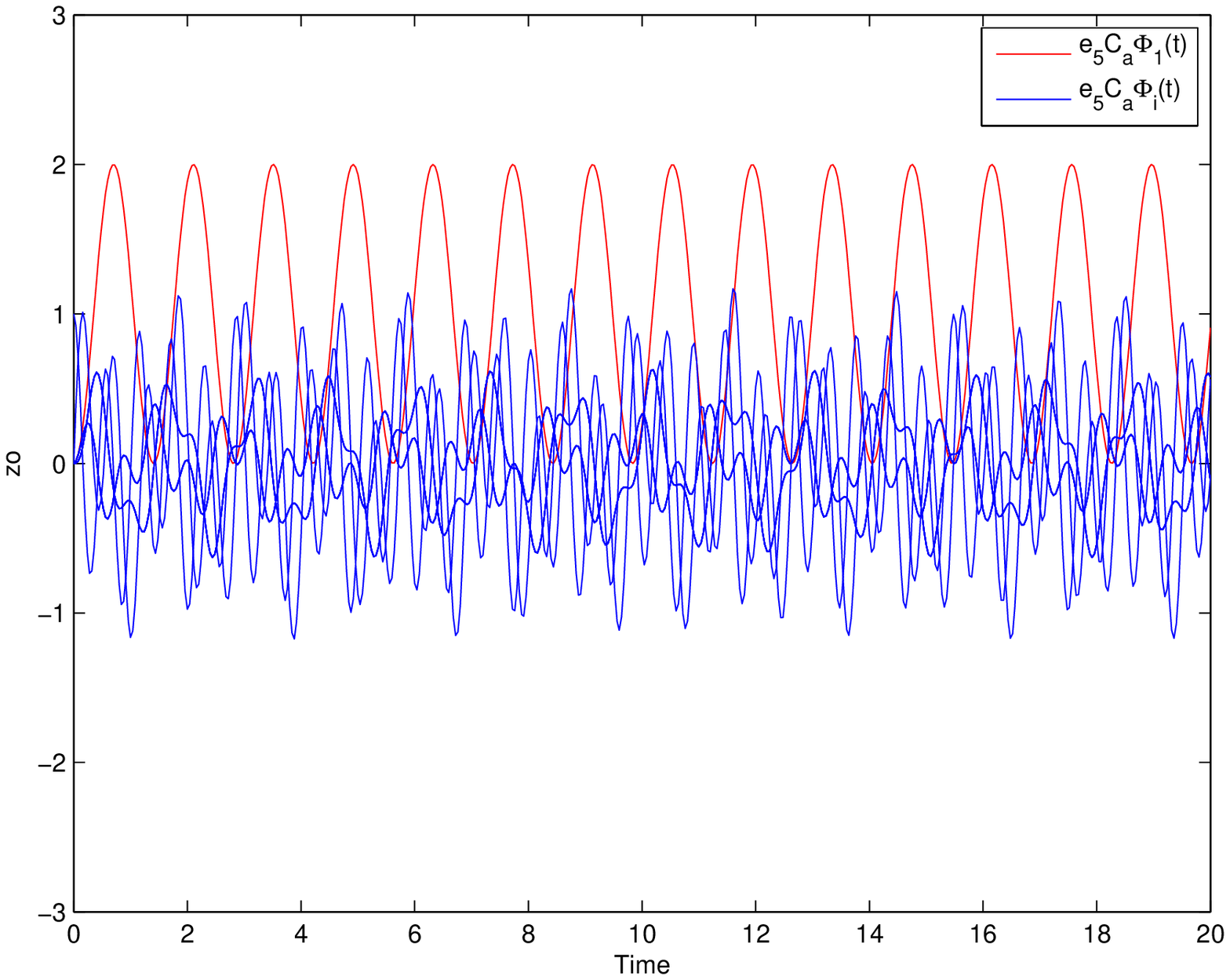}}%
\qquad
\subfloat[][]{\includegraphics[width=3.9cm]{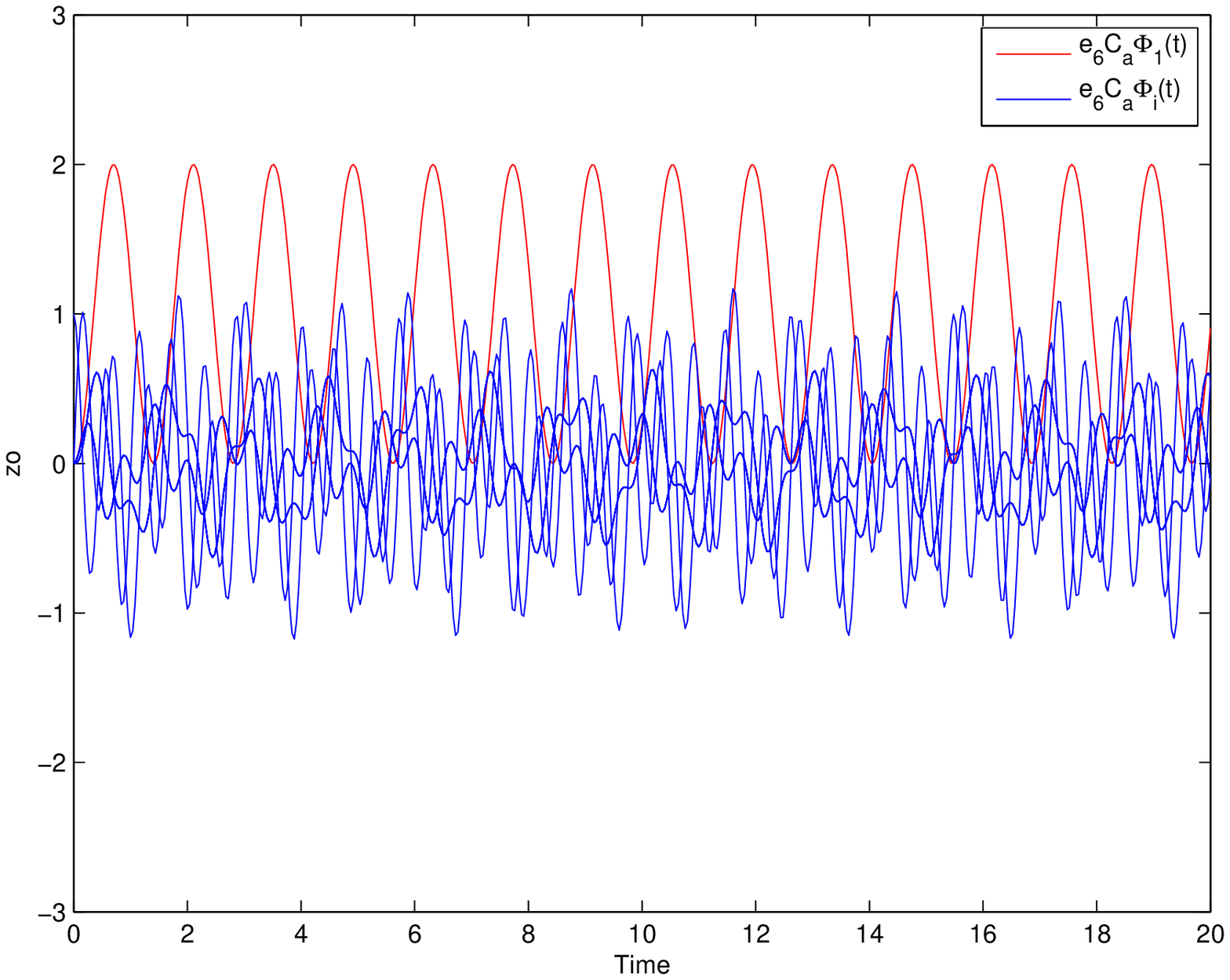}}\\
\caption{Coefficients defining (a) $z_p(t)$, (b) $z_{o1}(t)$, (c) $z_{o2}(t)$, (d) $z_{o3}(t)$, (e) $z_{o4}(t)$, and (f) $z_{o5}(t)$.}%
\label{F5}%
\end{figure}

From this figure, we can see that $e_1C_a\Phi_1(t)\equiv 1$ and  $e_1C_a\Phi_{2}(t)\equiv 0$, $e_1C_a\Phi_{2}(t)\equiv 0$, $\ldots$, $e_1C_a\Phi_{2N+2}(t)\equiv 0$, and $z_p(t)$ will remain constant at $z_p(0)$ for all $t \geq 0$.

We now consider the output variables of the  quantum observer network $z_{oi}(t)$ for $i=1,2,\ldots,N$ which are given by
\[
z_{oi}(t) = \sum_{j=1}^{2N+1}e_{i+1}C_a\Phi_j(t) x_{aj}(0)
\]
where $e_{i+1}$ is the $(i+1)$th unit vector in the standard basis for $\rbb^{N+1}$. We plot each of the quantities  
$e_{i+1}C_a\Phi_1(t),e_{i+1}C_a\Phi_2(t),\ldots,e_{i+1}C_a\Phi_{2N+2}(t)$ in Figures  \ref{F5}(b) - \ref{F5}(f).

To illustrate the time average convergence property of the quantum observer (\ref{average_convergence}), we now plot the quantities
$\frac{1}{T}\int_0^Te_{i+1}C_a\Phi_1(t)dt$, $\frac{1}{T}\int_0^Te_{i+1}C_a\Phi_2(t)dt$, $\ldots$, $\frac{1}{T}\int_0^Te_{i+1}C_a\Phi_{2N+2}(t)dt$
for $i=1,2,\ldots,N$ in Figures \ref{F6}(a)-\ref{F6}(e). These quantities determine the averaged value of the $i$th observer output
\[
z_{oi}^{ave}(T) = \frac{1}{T}\int_0^T\sum_{j=1}^{2N+1}e_{i+1}C_a\Phi_j(t) x_{aj}(0)dt
\]
for $i=1,2,\ldots,N$. 
\begin{figure}%
\centering
\subfloat[][]{\includegraphics[width=3.9cm]{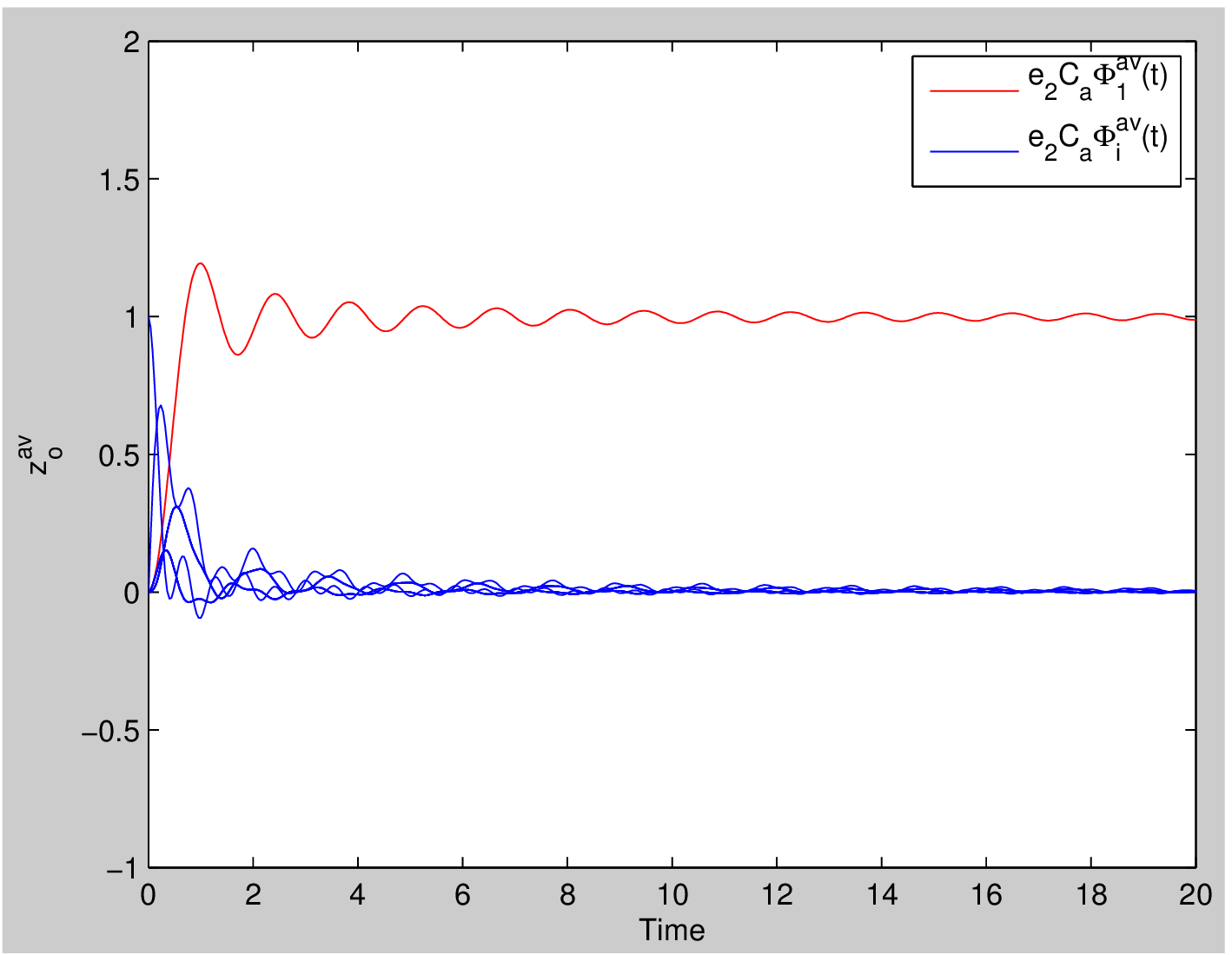}}%
\qquad
\subfloat[][]{\includegraphics[width=3.9cm]{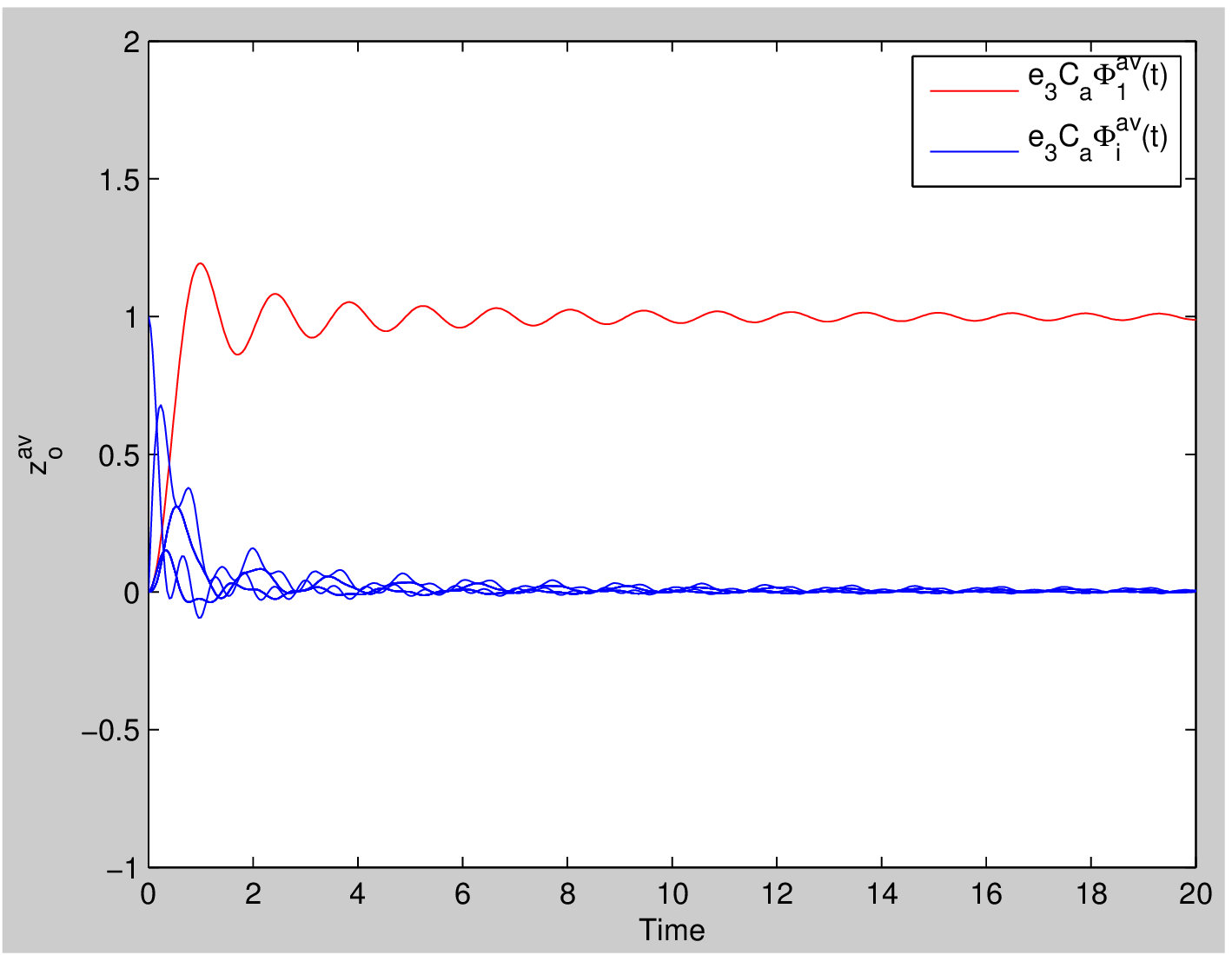}}\\
\subfloat[][]{\includegraphics[width=3.9cm]{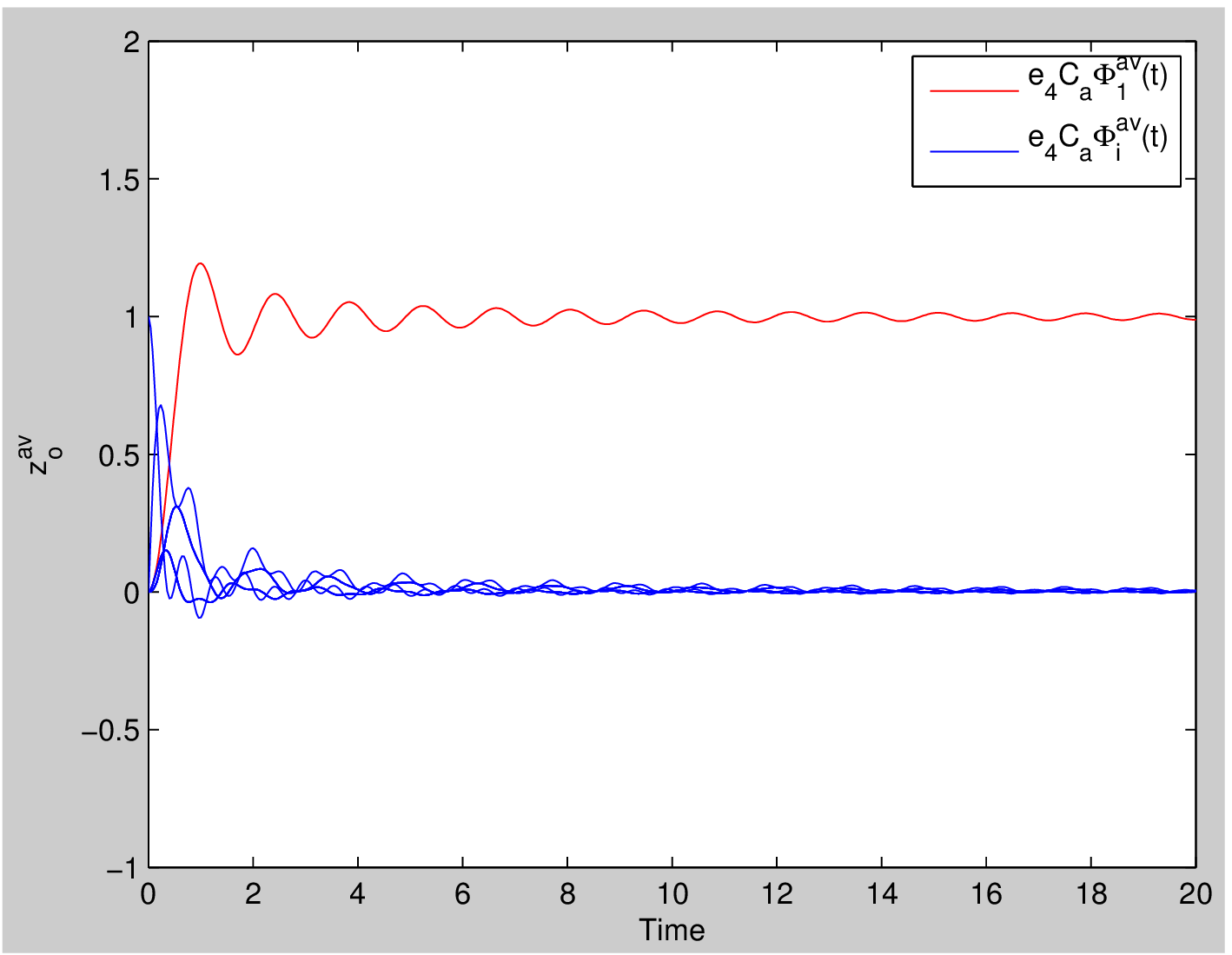}}%
\qquad
\subfloat[][]{\includegraphics[width=3.9cm]{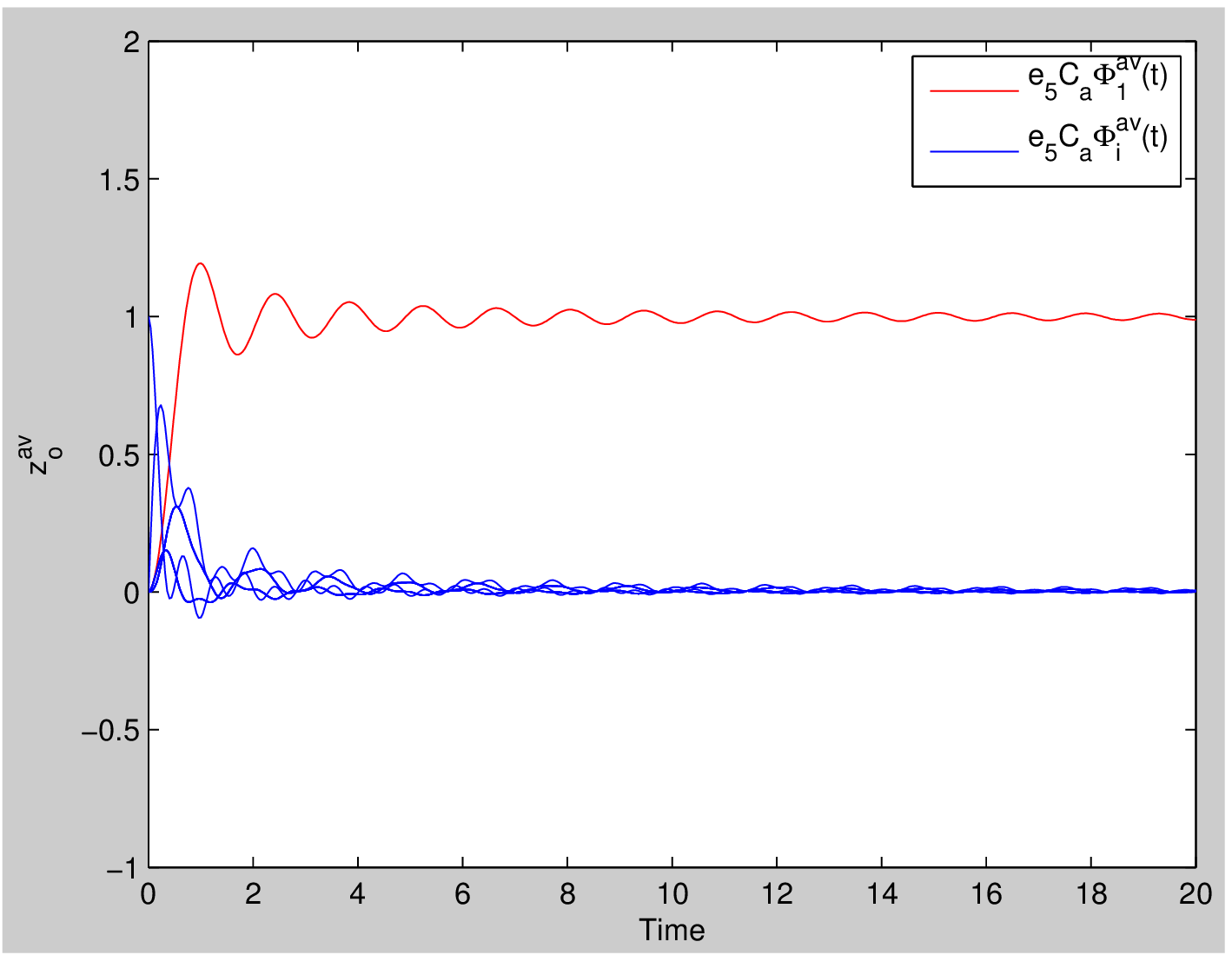}}\\
\subfloat[][]{\includegraphics[width=7cm]{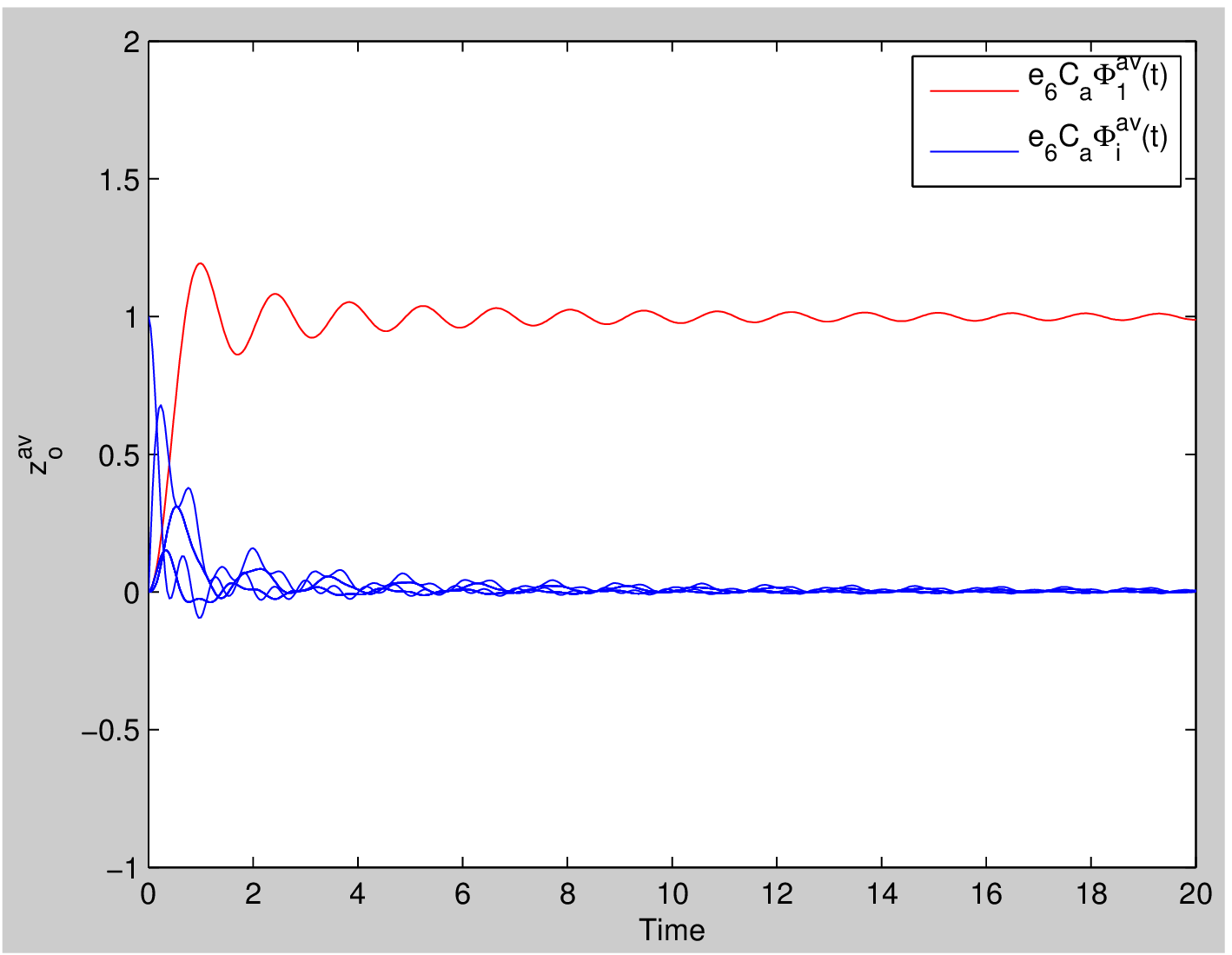}}
\caption{Coefficients defining the time average of (a)  $z_{o1}(t)$, (b) $z_{o2}(t)$, (c) $z_{o3}(t)$, (d) $z_{o4}(t)$, and (e) $z_{o5}(t)$.}%
\label{F6}%
 \end{figure}
From these figures, we can see that for each $i=1,2,\ldots,N$, the time average of $z_{oi}(t)$  converges to $z_p(0)$ as $t \rightarrow \infty$. That is, the  quantum observer network reaches a time averaged consensus corresponding to the output of the quantum plant which is to be estimated. 


\end{document}